\documentclass[10pt,twocolumn,twoside]{IEEEtran}
\IEEEoverridecommandlockouts

\usepackage{cite}
\usepackage{amsmath,amssymb,amsfonts}
\usepackage{graphicx}
\usepackage{textcomp}
\usepackage{balance}
\usepackage{xcolor}
\usepackage[utf8]{inputenc}

\usepackage{booktabs}
\usepackage{amsthm}
\newtheorem{theorem}{Theorem}[section]
\newtheorem{lemma}[theorem]{Lemma}

\usepackage{multirow}
\usepackage{makecell}
\usepackage{float}
\usepackage[markup=startstop]{changes}
\usepackage{threeparttable}

\usepackage{algorithm}
\usepackage{algpseudocode}

\usepackage{tabularx}
\usepackage{array}
\usepackage{booktabs}
\usepackage{makecell}

\makeatletter
\long\def\@makefntext#1{%
  \parindent 1em%
  \noindent\hbox{\@makefnmark}\hspace{0.5em}#1}
\makeatother

\let\oldthebibliography\thebibliography
\renewcommand{\thebibliography}[1]{%
  \oldthebibliography{#1}%
  \setlength{\itemsep}{0pt}%
  \setlength{\parsep}{0pt}%
}

\def\BibTeX{{\rm B\kern-.05em{\sc i\kern-.025em b}\kern-.08em
    T\kern-.1667em\lower.7ex\hbox{E}\kern-.125emX}}
\usepackage{lipsum}

\title{\LARGE \bf
Dynamic Storage Operation Under Uncertainty and the Reliability Externality:
Implications for Capacity Investments
}

\author{Daniel Shen$^\dagger$, Marija Ilic$^\dagger$, John Parsons$^\ddagger$
\vspace{-0.5cm}
\thanks{
    $^\dagger$ Department of Electrical
    Engineering and Computer Science, Massachusetts Institute of Technology,
    Cambridge, MA, USA. Email: {\tt
(oski, ilic)@mit.edu}}
\thanks{$^\ddagger$
Center for Energy and Environmental Policy Research, MIT, Cambridge, MA, USA.
Email: {\tt jparsons@mit.edu}}

\thanks{
This material is based upon work supported by the MIT Energy Initiative's
Future Energy Systems Center.}}

\begin{document}
\begingroup
\allowdisplaybreaks

\maketitle

\begin{abstract}
Energy storage is increasingly relied upon to meet short-term demand
uncertainties from renewable variability and electrification. Unlike
conventional generators, storage's contribution to reliability is
policy-dependent and balances near-term arbitrage against future scarcity risk.
We study how demand uncertainty alters such dynamic storage operation
and how these operating decisions propagate into long-run investment outcomes.
We formulate storage operation as an average-cost Markov decision process and
embed the resulting stationary policies into a stylized capacity expansion
framework. Demand uncertainty induces a precautionary storage policy which
hedges against stochastic scarcity, leading to materially different
post-storage demand distributions relative to perfect-foresight benchmarks. We
additionally demonstrate that the reliability externality characteristic of
electricity markets interacts with uncertainty in a manner that uniquely
distorts both storage operation and investment.
\end{abstract}


\section{Introduction} \label{sec:intro}

Energy storage systems, and in particular grid-scale battery energy storage,
are playing an increasing role in resource adequacy. Unlike conventional
generators, the ability of storage to meet demand is dependent on the amount of
stored energy available at that moment: effective operation depends on
appropriately scheduling charge and discharge decisions in anticipation of
future demand conditions. Planning a least-cost resource mix with storage as a
central technology requires jointly accounting for the subsequent operational
decisions of storage resources.

This coupling is increasingly relevant as short-term demand uncertainty grows
with higher penetrations of variable renewable generation and the
electrification of end uses. Under uncertainty, storage operation takes on a
dynamic,\footnote{Throughout this paper, ``dynamic'' refers to operational
optimization under short-term uncertainty in the energy market context,
not sub-hourly fast frequency reserves or dynamic stability.}
precautionary character whereby energy is reserved for a
\textit{potential}
period of high-demand (or high-price) events in the future~
\cite{durmazPrecautionaryStorageElectricity2016}.
The expansion in installed battery capacity over the past decade has
subsequently spurred a rich literature studying such dynamic scheduling of
battery storage under uncertainties. Such approaches as in
\cite{megelStochasticDualDynamic2015,
jiangOptimalHourAheadBidding2015, xuOperationalValuationEnergy2020,
zhengArbitragingVariableEfficiency2022, balakinDynamicTradingStrategies2025,
guerraRobustScalableDispatch2025}
adopt a stochastic
programming perspective on operation,
while a closely related strand of literature links such dynamic storage
operation to long-term investment decisions to encompass the feedback between
these two problems \cite{geskeOptimalStorageInvestment2020,
pengRenewableFlexibleStorage2024, holeCapacityPlanningRenewable2025,
schmidtLongdurationStorageWeather2025,
biggarTheoryStoragePower2026}.

Much of the capacity expansion literature, however,
is formulated from a system-planning
perspective and does not analyze the reliability externality
that besets investment incentives in liberalized electricity markets
\cite{franka.wolakLongTermResourceAdequacy2021}.
Grid reliability is a non-excludable public good
and short-run prices are often capped to prevent the exercise of market power.
Thus, energy prices do not fully reflect the value of installed capacity and
revenues earned through energy markets may be insufficient to support adequate
investment~\cite{joskowCapacityPaymentsImperfect2008, suskiMissingMoneyMarketBased2025}.%
\footnote{Regulatory price caps are the most prominent contributor to
this phenomenon. Even in the absence of explicit caps, however, prices may fail
to reach levels commensurate with the value of lost load since there is a lack
of individualized penalties for failing to procure capacity for reliability.}
Many electricity markets address this externality
through uniform payments based on pre-calculated measures of firm
capacity.
\footnote{Four U.S. markets (PJM, ISO-NE, NYISO, and MISO) utilize
capacity markets to pay suppliers for commitments to meet future demand
~\cite{officeofenergypolicyandinnovationEnergyPrimerHandbook2023}.}
However, such capacity payment mechanisms were primarily designed
around thermal generators and may insufficiently compensate resources with
state-contingent supply such as storage and renewables
\cite{franka.wolakLongTermResourceAdequacy2021}. While availability
obligations such as PJM's Capacity Performance Mechanism
can mitigate these shortcomings, they are also infrequently declared and
subjective in their assessment of penalties
\cite{waltergrafPJMsCapacityPerformance2025}.

In this work, we adopt a stochastic control framework to study how demand
uncertainty a) shapes the operation and investment of energy storage in
electricity systems, and b) interacts with the reliability
externality. We model storage operation as an infinite-horizon average-cost
Markov decision process and embed the resulting stationary operating policies
within a capacity expansion model.
This formulation allows us to propagate
short-run operational uncertainty into long-run equilibrium
investment decisions. In particular, our work addresses the following
questions:

\begin{enumerate}
    \item \textit{How does short-term demand uncertainty alter the operating
    policy of energy storage, relative to when the demand is known with perfect
    foresight?}
    \item \textit{How do uncertainty-driven storage operating policies map into
    long-run equilibrium investment decisions for storage and conventional
    generation capacity?}
    \item \textit{How does the reliability externality in electricity markets
    distort storage operation and capacity investment under uncertainty, and
    why are conventional capacity mechanisms insufficient to correct these
    distortions?}
\end{enumerate}

To this end, we develop a stylized capacity expansion model that embeds a
dynamic storage operating model. Although our model is deliberately kept
simple, it retains the structure needed to capture interactions between
uncertainty, storage operation, and investment incentives. In particular, our
contributions are as follows:

\begin{itemize}
    \item We characterize how demand uncertainty induces a precautionary
    storage policy distinct from perfect-foresight
    operation. This ``precautionary storage'' reserves energy to hedge against
    stochastic scarcity events. We then embed the precautionary policy into a
    capacity expansion model to compare equilibrium investment outcomes with
    a setup with perfect foresight.
    \item We demonstrate the reliability externality distorts the operating
    policies in a manner that qualitatively differs from conventional
    generators due to the dynamic nature of storage operation. We additionally
    model how these operational distortions propagate to long-run planning
    outcomes. Although our results do not find a significant impact on storage
    investment, we posit the externality should still lead to suboptimal
    storage capacity that cannot be corrected by standard capacity payments.
\end{itemize}

The rest of the paper is organized as follows. Section \ref{sec:model}
introduces our operation and investment model. Section
\ref{sec:operation_and_investment} illustrates the operating and investment
characteristics of storage under uncertain net demand. Section
\ref{sec:missing_money} analyzes how the reliability externality in electricity
markets uniquely distorts outcomes for storage. Finally, Section
\ref{sec:conclusion} concludes.

\section{Electricity System Model} \label{sec:model}

\newcommand{\capacities}{\ensuremath{\mathbf{x}}}
\newcommand{\operations}{\ensuremath{\mathbf{y}}}

We utilize a mean-reverting model of inelastic electricity demand combined with
a stylized capacity expansion model that embeds a storage operating model.
Storage operation is based on a Markov decision process (MDP) with
probabilistic foresight of future net demand conditions. While our model omits
consideration of transmission constraints, ramp rates, and co-optimization of
investments in variable renewable energy resources, we are primarily concerned
with the basic interaction between storage operation and net demand uncertainty
and so opt for this simplified framework.

\subsection{Demand Model} \label{sec:model-ou}

Let the net demand $\ell$ be a Markov chain where the transition
probabilities $P_{\mathcal{L}}(\ell' \mid \ell)$ are calculated using an
Ornstein-Uhlenbeck (OU) process. This treatment allows us to both capture the
mean-reverting nature of net electricity demand and represent uncertainty in
future demand.
The OU process is a continuous-time stochastic
process that exhibits mean-reverting behavior. The
process is described by the stochastic differential equation:

\begin{equation} \label{eq:ornstein-uhlenbeck}
    dx_t = \theta (\mu - x_t) dt + \sigma dW_t,
\end{equation}
where $x$ is the state variable (in our case, the net demand $\ell$),
$\theta$ governs the drift speed of mean reversion toward the long-term mean
$\mu$, $\sigma$ is the volatility parameter, and $W_t$ is a standard Wiener
process.

Given a time interval $\Delta t$ and a starting state $x_t$, the conditional distribution of
$x_{t+\Delta t}$ is normally distributed with mean and variance:
\begin{align}
    \mathbb{E}[x_{t+\Delta t} \mid x_t] & = x_t e^{-\theta \Delta t} + \mu \left(1 - e^{-\theta \Delta t}\right) \\
    \text{Var}[x_{t+\Delta t} \mid x_t] & = \frac{\sigma^2}{2\theta} \left(1 - e^{-2\theta \Delta t}\right).
\end{align}

Additionally, the long-term stationary distribution of the OU process is Gaussian
with long-term mean $\mu$ and variance $\frac{\sigma^2}{2\theta}$.

For our discrete-time MDP framework, we discretize the demand into
equally-spaced levels and calculate transition probabilities between levels by
integrating the conditional distributions of~\eqref{eq:ornstein-uhlenbeck}.

\subsection{Dynamic Storage Operating Model} \label{sec:operating-model}

Uncertainty of future net demand motivates
reserving stored energy to hedge against potential demand states.
Our broader analysis concerns the long-run equilibrium of operating and
investment decisions; hence, we formulate the storage operating problem as an
average-cost MDP where the system operator optimizes the dispatch to minimize
expected costs over an infinite horizon.

The system has conventional generators $g \in G$, each with power
capacity $K_g$ and constant variable cost $c_g$. There is a single storage
technology with power capacity $K_s$, stored energy capacity $E$, and round-trip
efficiency $\eta$. If insufficient supply is available to fully meet
net load, some load can be unserved / shed at the cost of the value of lost load
$c_{shed}$. Finally, when the net demand is negative
(\textit{i.e.} excess renewable generation), surplus energy can be curtailed at no cost.

In each interval the system operator determines the dispatch based on the
system state $S \in \mathcal{S} = \mathcal{E} \times \mathcal{L}$. This state
is comprised of $e$, the stored energy available for use over the interval, and
$\ell$, the demand to be met during the interval. Based on the state, the
operator subsequently chooses the mix of generation and storage
dispatch which minimizes total costs in expectation. Let the dispatch
decision affecting the interval be represented by the action $a = (p_{s}^{ch},
p_{s}^{dis}, p_{g}, p_{shed}, p_{curt})$. This vector captures the charge and
discharge rates of storage, output of conventional
generators, unserved energy, and curtailed surplus, respectively.
\footnote{The injections $p_s^{(ch,dis)}$ capture
the rate of change of stored energy; the charging load seen by the grid is
adjusted by the efficiency factor $\eta$.}

Given fixed storage and generation capacities, the set of feasible actions
$\mathcal{A}(S)$ includes all actions that satisfy the operating constraints
given the current state. Specifically, $a \in \mathcal{A}(S)$ if $a$ satisfies:

\begin{subequations} \label{eq:operational_constraints}
\begin{align}
& \sum_{g} p_{g} + p_{s}^{dis} + p_{shed} = \ell + \frac{1}{\eta} p_{s}^{ch} + p_{curt} \label{eq:gencost} \\
& 0 \leq p_{g} \leq K_g \quad \forall g \label{eq:gencap} \\
& 0 \leq p_{shed} \label{eq:shed} \\
& 0 \leq p_{curt} \label{eq:curt} \\
& 0 \leq \frac{1}{\eta} p_{s}^{ch} \leq K_s \label{eq:storage-ch} \\
& 0 \leq p_{s}^{dis} \leq K_s \label{eq:storage-dis} \\
& p_{s}^{ch} p_{s}^{dis} = 0 \label{eq:nonsimultaneous} \\
& e + p_{s}^{ch} - p_{s}^{dis} = e' \label{eq:storage-energy-trajectory} \\
& 0 \leq e' \leq E. \label{eq:soc-cap}
\end{align}
\end{subequations}
Where~\eqref{eq:gencost} is the power balance constraint,~\eqref{eq:gencap}
enforces generation capacity limits,~\eqref{eq:shed} enforces non-negativity of
load shedding,~\eqref{eq:curt} enforces non-negativity of curtailment,
\eqref{eq:storage-ch} and~\eqref{eq:storage-dis} enforce grid-side storage
power capacity limits,~\eqref{eq:nonsimultaneous} prevents simultaneous
charging and discharging,
\eqref{eq:storage-energy-trajectory} describes
the evolution of stored energy from state $e$ to $e'$ as a result of charging
and discharging, and~\eqref{eq:soc-cap} enforces stored energy capacity
limits.\footnote{
    Constraints~\eqref{eq:gencost}--\eqref{eq:soc-cap}
    are practically modeled by removing points
    from the discrete action set instead of encoding them as constraints in
    a solver, thus avoiding the nonlinearity of \eqref{eq:nonsimultaneous}.
}

Let $C(\ell, a)$ represent the system operating costs, which include both
generation costs and load shedding costs:

\begin{equation}
    C(\ell, a) = c_{shed} p_{shed}  +
    \sum_{g \in \mathcal{G}} c_{g} p_{g}.
\end{equation}

This infinite-horizon operational problem is commonly cast as a Bellman
equation and solved with dynamic programming techniques. However, since our
analysis is concerned with the long-run impact of storage operation on capacity
investments \textit{in equilibrium}, we instead opt to formulate the operating
problem using the ``average cost'' criteria,
\cite{arapostathisDiscreteTimeControlledMarkov1993} where the objective is
to minimize the long-run expected average cost per time step. Specifically,
given a fixed operating policy $\pi: \mathcal{S} \to \mathcal{A}$, the long-run
average cost under the induced stationary distribution is given by:

\begin{equation}
    J(\pi) =
    \limsup_{T \to \infty} \frac{1}{T}
    \mathbb{E}_\pi\!\left[
    \sum_{t=0}^{T-1} C(\ell_t, a_t)
    \right].
\end{equation}

We next establish the existence of a stationary distribution for the state
process induced by the storage and load model:

\begin{lemma}[Existence of a stationary distribution]
\label{lem:stationary_distribution}
Under any stationary policy $\pi$, the controlled state process
$S=(e,\ell)$ induced by the stored energy constraint
\textit{\eqref{eq:storage-energy-trajectory}} admits a stationary distribution
$\rho$ which is unique on the recurrent class induced by $\pi$.
\end{lemma}

\begin{proof}
The exogenous net-load process $\{\ell\}$ is obtained by discretizing a
mean-reverting Ornstein--Uhlenbeck process on a finite grid; consequently, the
net-load transition matrix $P_{\mathcal L}(\ell' \mid \ell)$ is irreducible and
aperiodic. Under a stationary policy $\pi$, the joint process $S=(e,\ell)$ is a
time-homogeneous Markov chain on the finite state space $\mathcal S=\mathcal
E\times\mathcal L$, with transitions determined by the storage state update in
\eqref{eq:storage-energy-trajectory} and $P_{\mathcal L}$. The policy $\pi$
may restrict the set of stored energy levels which are visited with positive
probability. \footnote{For example, if the stored energy capacity is much
larger than the net load peaks, the upper bound on stored energy may be slack
and the optimal policy may never utilize the upper levels of the storage.} Let
$\mathcal S^\pi\subseteq\mathcal S$ denote the corresponding recurrent class.
The induced chain is ergodic on its recurrent class
$\mathcal S^\pi$, and therefore admits a unique stationary distribution
$\rho$ supported on $\mathcal S^\pi$.
\end{proof}

The stationary distribution and accompanying policy can be computed via dynamic
programming methods (\textit{e.g.}, value iteration), but our work's overall analysis
focuses on the long-run equilibrium behavior of the system. Hence, we opt to
directly find $\rho$ and $\pi$ using the occupancy measure formulation.

\subsubsection{Occupancy Measure LP Formulation} \label{sec:occupancy-lp}

To compute $\rho$, we can utilize the occupancy-measure
formulation \cite{manneLinearProgrammingSequential1960}
which seeks $\rho$ as part of the solutions to a linear
program constructed around the \emph{occupancy measure} that minimizes the
expected average cost.
The occupancy measure $q(S,a)$ represents the long-run fraction of time the
system spends in each state-action pair; the occupancy which minimizes
system operating costs can be found with the following linear program (LP):

\begin{subequations} \label{eq:occupancy_lp}
\begin{align}
    \min_{q}\quad \label{eq:occupancy_lp_obj} &
    \sum_{S}\sum_{a\in\mathcal{A}(S)} q(S,a)\,C(\ell,a) \\
    \text{s.t.}\quad &
    q(S,a)\ge 0, \quad \forall S\in\mathcal{S},\;\forall a\in\mathcal{A}(S), \label{eq:occupancy_nonneg} \\
    & \sum_{S} \sum_{a} q(S,a)=1, \label{eq:total_probability}\\
    & \sum_{a} q(S, a) = \sum_{S'} \sum_{a'} q(S', a')\,P(S \mid S', a')
    \quad \forall S\in\mathcal{S}. \label{eq:flow_conservation}
\end{align}
\end{subequations}

Here, the objective~\eqref{eq:occupancy_lp_obj}
minimizes the expected average system cost per
time step. The first two constraints~\eqref{eq:occupancy_nonneg} and
\eqref{eq:total_probability} ensure that $q$ is a valid probability
distribution and the last constraint~\eqref{eq:flow_conservation} enforces flow
conservation for each state $S$
such that the probability of entering state
$S$ in any interval equals the probability of leaving it.

The optimal $q^\star$ induces a distribution of actions:
\begin{equation}
    \mathbb{P}[a \mid S]
    =
    \frac{q^\star(S,a)}{\sum_{a'} q^\star(S,a')}, \quad a' \in \mathcal{A}(S)
\end{equation}
\noindent on its support. Because the stored energy transition is
deterministic conditional on the action (\textit{i.e.} the next state of charge is
uniquely determined by $(S,a)$), for each state $S$ all mass is assigned to a
single action and the action taken in state $S$ is deterministic; thus the
optimal policy $\pi^*$ is also deterministic.

By marginalization, the state occupancy satisfies:
\begin{equation}
    \rho^{\star}(S) = \sum_a q^\star(S,a),
\end{equation}
thus recovering the stationary distribution.

Solving~\eqref{eq:occupancy_lp_obj} as a LP involves discretizing the load,
stored energy, and action spaces. In particular, the discretization step of
$p_s^{ch}$ and $p_s^{dis}$ must match the discretization step of stored energy
levels so that conservation of energy is maintained.

\subsection{Capacity Expansion Model under Dynamic Storage Operation}
\label{sec:capacity_expansion}

Storage operation and generation capacity choices are mutually dependent: the
optimal storage policy depends on the cost spread of generators, while optimal
generation capacities depend on the arbitrage provided by storage.
Building on~\cite{geskeOptimalStorageInvestment2020}, we address this interaction
via an iterative fixed-point algorithm that alternates between determining
storage operation and adjusting generation capacities. We additionally extend the
aforementioned approach by introducing a capacity update rule for storage based
on a long-run zero-profit condition.

Let $\mathbf{K}$ denote the vector of installed capacities, consisting of
generation capacities $\mathbf{K}_{\mathcal G}$ and storage power capacity
$K_s$. Storage is assumed to have a fixed duration $\tau$, so that its
energy capacity satisfies $E = \tau K_s$.
For any fixed capacity vector $\mathbf{K}$, the storage operating problem
(Section~\ref{sec:occupancy-lp}) admits an optimal policy
$\pi^\star$ with associated occupancy measure $q^{\mathbf K}(S,a)$ and
induced stationary state distribution $\rho^{\mathbf K}(S)$.

Finally, let $G(S,a)$ denote the \textit{post-storage net demand} which
must be balanced by conventional generation, curtailment, and load shedding:

\begin{equation} \label{eq:post-storage-demand}
    G(S,a) = \ell + \frac{1}{\eta} p_s^{ch} - p_s^{dis}.
\end{equation}

\paragraph{Generation Capacity Update Under Fixed Storage Policy}

Given $\mathbf K$, the stationary distribution of the post-storage net demand
$G = G(S,a)$ is given by:

\begin{equation} \label{eq:pdf-post-storage-demand}
\mathbb{P}\left[ G = x \right]
=
\sum_{S \in \mathcal S}
\sum_{a \in \mathcal A(S)}
q^{\mathbf K}(S,a)\,
\mathbf 1\!\left\{ G(S,a) = x \right\}.
\end{equation}
Let $F_G(x) = \mathbb{P}\{ G \le x \}$ denote the corresponding cumulative
distribution function with an inverse function $F_G^{-1}(u) = \inf\{x :
F_G(x) \ge u\}$. The associated post-storage load-duration curve is:

\begin{equation} \label{eq:post-storage-duration-curve}
D_G(h) = F_G^{-1}(1-h), \qquad h \in [0,1].
\end{equation}
This curve is the direct analogy to the canonical load-duration curve, with
post-storage net demand replacing net load. As a result, standard
screening-curve methods can be applied to $D_G$ to update
$\mathbf{K}_{\mathcal{G}}$ based on each technology's corresponding fixed costs
$I_g$.

\paragraph{Storage Capacity Update With Zero-Profit Condition}

We update storage capacity using a long-run zero-profit condition computed
under the stationary distribution induced by optimal operation.
For each state--action pair $(S,a)$ with positive occupancy, define the
\emph{statewise marginal price}:
\begin{equation}
\lambda(S,a)
\;\triangleq\;
\max\!\Big\{
\max_{g:\,p_g(S,a)>0} c_g,\;
c_{\mathrm{shed}}\cdot\mathbf{1}\{p_{\mathrm{shed}}(S,a)>0\}
\Big\},
\end{equation}
\textit{i.e.}, the variable cost of the most expensive generator dispatched in that
period, or VoLL if load shedding occurs. The corresponding storage operating
surplus in state $S$ is:

\begin{equation}
\phi_s(S,a) =
    \lambda(S,a) \left(p_s^{{dis}} - \frac{1}{\eta} p_s^{{ch}}\right).
\end{equation}

The expected annual operating surplus of storage over the planning horizon is:
\begin{equation}
\Phi_s(\mathbf{K})
=
T
\sum_{S\in\mathcal{S}}\sum_{a\in\mathcal{A}(S)}
q^{\mathbf{K}}(S,a)\,\phi_s(S,a),
\end{equation}
where $T$ is the number of operating intervals per year.

Let $I_s^\tau$ denote the annualized fixed cost of a storage
technology with duration $\tau$. Storage capacity is updated via a proportional
stepsize rule until surplus and investment break even. Denote
$\mathbf{K}_\mathcal{G}^{(m+1)}$ as the optimal generation capacities
from the generation capacity update procedure based on the fixed
storage capacity $K_s^{(m)}$. The storage update is:

\begin{equation}
K_s^{(m+1)}
=
\Big[ K_s^{(m)} +
    \gamma_s \big(\Phi_s(\mathbf{K}_\mathcal{G}^{(m+1)}, K_s^{(m)}) -
    I_s^\tau K_s^{(m)} \big) \Big]_+,
\label{eq:storage_stepsize_update}
\end{equation}
where $\gamma_s>0$ is a stepsize parameter and $[\cdot]_+$ enforces
nonnegativity. At convergence,
\(
\Phi_s(\mathbf{K}^\star)
=
I_s^\tau K_s^\star.
\)

\paragraph{Fixed-Point Capacity Expansion Algorithm}
For a fixed storage duration $\tau$, the equilibrium $\mathbf{K}$ is
determined by the following fixed-point algorithm:

\begin{samepage}
\begin{enumerate}
\item Initialize storage and generation capacities $\mathbf{K}^{(0)}$.
\item For outer iteration $m=0,1,2,\ldots$:
\begin{enumerate}
    \item \emph{Generation update:} Hold storage capacity $K_s^{(m)}$ fixed
        and update $\mathbf{K}_{\mathcal{G}}^{(m)}$ to new capacities
        $\mathbf{K}_{\mathcal{G}}^{(m + 1)}$ by alternating between
        the operating LP and screening curve until a fixed point between
        storage operation and generation capacities is reached.
    \item \emph{Storage update:} compute
    $\Phi_s(\mathbf{K}_\mathcal{G}^{(m+1)}, K_s^{(m)})$ and update
    $K_s^{(m+1)}$ using~\eqref{eq:storage_stepsize_update},
    $E^{(m+1)}=\tau K_s^{(m+1)}$.
\end{enumerate}
\end{enumerate}
\end{samepage}

The algorithm terminates when both storage and generation capacities converge
within a prescribed tolerance $\varepsilon$:
\[
\lvert K_s^{(m+1)} - K_s^{(m)} \rvert \le \varepsilon_s
\quad \text{and} \quad
\bigl\|
\mathbf{K}_{\mathcal G}^{(m+1)} -
\mathbf{K}_{\mathcal G}^{(m)}
\bigr\| \le \varepsilon_{\mathcal G}.
\]

\section{Storage Operation and Investment}
\label{sec:operation_and_investment}

In this section, we study the operation and investment decisions
in a stylized electricity system with uncertain net demand, two generation
technologies, and a single storage technology. The analysis is intended to
highlight the general implications of dynamic storage utilization for long-run
planning, rather than to represent any specific real-world system.

\subsection{Example Electricity System}

We construct our system based on a hypothetical future of a
high-renewable, highly electrified grid where there is significant net demand
uncertainty and frequent periods of negative net demand. To this end,
we utilize an Ornstein-Uhlenbeck process (Section \ref{sec:model-ou})
whose drift and volatility are calibrated to 2024 net demand data from
the California ISO (CAISO) balancing area. We exogenously set the long-term mean $\mu$
of the process to 6.5~GW, which creates a net demand
that is negative approximately 15\% of the time.
\footnote{See Table~\ref{tab:demand_parameters} in the Appendix for further discussion on the
          parameters.}

The system planner may invest in three technologies: a baseload generator, a
peaking generator, and an energy storage technology with a fixed four-hour
duration. The operating problem is solved on an hourly basis.
Numerical values for the demand process and technology cost
parameters are provided in Appendix~\ref{appendix:system-values}.

\subsection{Operating Policy and Investment Decisions under Demand Uncertainty}

We first examine the operating policy of storage under demand uncertainty
when the system capacities are fixed.
Figure \ref{fig:load_duration_curve} shows the canonical (pre-storage) net
load-duration curve
of the system along with two post-storage load-duration curves.
The post-storage curve represents the distribution of net demand which must
be met by conventional generation or load shedding after the storage action is
taken. Charging shifts the post-storage curve above the pre-storage curve
and discharging vice-versa.

Under dynamic operation, the storage unit primarily
charges during two regions: i) when net demand is below 10~GW and baseload
generation is marginal, and ii) when net demand is within 15--23~GW and
peaker generation is marginal. The second condition corresponds
to a ``precautionary storage'' operation whereby demand uncertainty induces
banking energy for the \textit{potential} of a period of
high demand.

In contrast, storage under perfect-foresight exhibits the same operating
behavior in the first region, but not in the second region. With full
information, it is possible to almost exclusively charge during baseload and
curtailment conditions to meet peaks where net demand exceeds peaker capacity,
hence the post-storage curve overlaps with the pre-storage in the 15--23~GW
region.

\begin{figure}[tb!]
    \centering
    \includegraphics[width=\columnwidth]{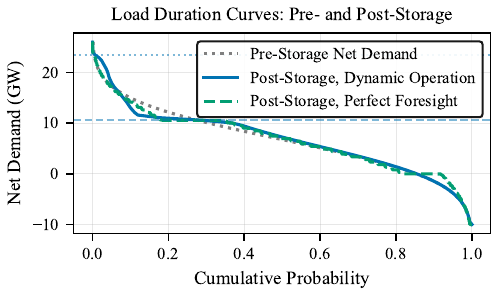}
    \caption{Load duration curves for i) net demand, ii) post-storage demand
             under dynamic (imperfect-foresight) operation, and iii)
             post-storage demand under perfect-foresight operation.
             Post-storage demand is higher than net demand when storage is
             charging, see \eqref{eq:post-storage-demand} for definition.
             System baseload and peaker capacity steps are denoted with
             horizontal dashed and dotted lines, respectively. Storage capacity
             is 4.0~GW\@.}
    \label{fig:load_duration_curve}
\end{figure}

Overall, the ``precautionary storage'' behavior under uncertainty results in a
more conservative operating policy that reserves a portion of energy for future
scarcity events rather than fully exploiting arbitrage opportunities.
Figure \ref{fig:soc_utilization} demonstrates this outcome:
the average stored energy levels at the beginning of each operating interval
(conditioned on the net demand) are everywhere higher in the dynamic operating
case than the perfect-foresight case. In the latter case, the operator
can more economically utilize stored energy by precisely timing discharges to
meet high-demand events.

Demand uncertainty induces a more conservative storage operating policy, which
substantially reduces the system value of storage.
At equilibrium, the perfect-foresight capacity mix includes 7.2~GW
of storage, whereas under demand uncertainty the mix only includes
4.0~GW (Table~\ref{tab:price_cap_capacity}).

The inability of storage to fully arbitrage scarcity
events under uncertainty shifts investment toward peaker generation
which has a greater availability to meet demand: peaker capacity
increases from 10.7~GW in the perfect-foresight case to
12.9~GW under uncertainty, and baseload capacity decreases from 11.1~GW to
10.6~GW\@.
Demand uncertainty biases the equilibrium capacity mix away from storage and
towards a greater total amount of conventional generation, whose availability
is less contingent on future state information.

\begin{figure}[tb!]
    \centering
    \includegraphics[width=\columnwidth]{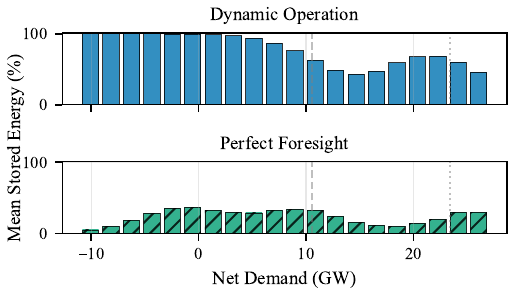}
    \caption{Expected level of stored energy conditioned on net demand state.
             Values for the perfect-foresight panel are constructed
             by sampling 10 years of hourly data from the net load process,
             computing the optimal dispatch, and then binning the state of
             charge by net demand level. Baseload and peaker capacity steps
             are denoted with vertical dashed and dotted lines, respectively.
             Capacities correspond to the first row of Table~\ref{tab:price_cap_capacity}.}
    \label{fig:soc_utilization}
\end{figure}

{\renewcommand{\arraystretch}{1.1}
\begin{table}[tb!]
\centering
\begin{threeparttable}
\caption{Capacity Investments (GW) Under Varying Price Caps}
\label{tab:price_cap_capacity}
\begin{tabular*}{\columnwidth}{@{\extracolsep{\fill}}lcccc@{}}
\toprule
\textbf{Technology} & \textbf{VoLL (\$/MWh)} &
\multicolumn{3}{c}{\textbf{Price Cap (\$/MWh)}} \\
\cmidrule(lr){2-2}\cmidrule(lr){3-5}
& 10{,}000 & 5{,}000 & 2{,}500 & 1{,}000 \\
\midrule
\multicolumn{5}{@{}l}{\textit{Dynamic Operation}} \\
Storage  & 4.0  & 3.9  & 4.0  & 3.9 \\
Baseload & 10.6 & 11.1 & 10.6 & 11.1 \\
Peaker   & 12.9 & 11.3 & 10.8 & 9.3 \\
\addlinespace[0.2em]
\textbf{Total} & \textbf{27.5} & \textbf{26.3} & \textbf{25.4} & \textbf{24.3} \\
\addlinespace[0.5em]
\multicolumn{5}{@{}l}{\textit{Perfect-Foresight Operation}} \\
Storage  & 7.2 & 6.8 & 6.6 & 6.2 \\
Baseload & 11.1 & 11.2 & 11.2 & 11.2 \\
Peaker   & 10.7  & 10.2 & 8.8  & 7.0 \\
\addlinespace[0.2em]
\textbf{Total} & \textbf{29.0} & \textbf{28.2} & \textbf{24.6} & \textbf{24.4} \\
\bottomrule
\end{tabular*}
\begin{tablenotes}[flushleft]
\footnotesize
\item Capacities for \textit{Dynamic Operation} are determined based on Section
~\ref{sec:capacity_expansion}, while capacities for \textit{Perfect-Foresight
    Operation} are determined by solving a deterministic capacity expansion over
    a 10-year hourly sample drawn from the demand transition matrix.
\end{tablenotes}
\end{threeparttable}
\end{table}
}

\section{Capacity Investments and the ``Missing Money'' Problem Unique to Storage}
\label{sec:missing_money}

\subsection{Missing Money Impacts on Capacity and Operation}

Grid reliability is a non-excludable public good: during rotating blackouts,
service interruptions cannot be targeted to individual customers and are
instead imposed broadly or along shared feeders. This reliability externality,
in conjunction with the regulatory price caps present in many markets, results
in an effective cap on short-run electricity prices below the social value of
lost load (VoLL). As a result, there is a disincentive for markets to invest in
the socially optimal level of capacity.

These distortions give rise to a ``missing money'' problem, whereby energy
market revenues alone cannot cover the fixed costs of resources needed to
maintain adequate reliability. Table~\ref{tab:price_cap_capacity} illustrates
the impact of this phenomenon on capacity investments by way of price cap which
truncates prices below the VoLL of \$10k/MWh. In both the dynamic and
perfect-foresight cases, peaker investment decreases in the presence of a price
cap, since the cap reduces revenues during scarcity events.

Storage investment decreases slightly in the dynamic setting over the modeled
range, from 4.0 GW at full VoLL to 3.9 GW at a \$1k/MWh cap. This pattern
should be interpreted with caution. The dynamic investment problem is
nonconvex, and the fixed-point algorithm used here does not guarantee
convergence to a global optimum. While the qualitative interactions between
uncertainty, operating policy, and investment incentives are robust, the
direction and magnitude of storage investment responses to price caps may
change under alternative solution methods or improved global search.

\begin{figure}[tb!]
    \centering
    \includegraphics[width=\columnwidth]{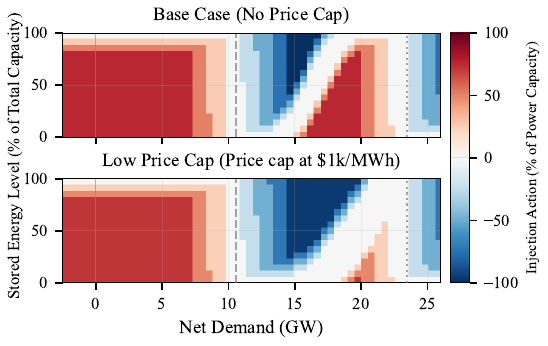}
    \caption{Impact of price cap on dynamic storage operation.
             Resource
             capacities are held constant at the VoLL capacities for
             \textit{Dynamic Operation} (first column, upper panel in
             Table~\ref{tab:price_cap_capacity}). Red = optimal to charge,
             blue = optimal to discharge, white = optimal to idle. System
             baseload and peaker steps are denoted with vertical dashed and
             dotted lines, respectively.}
    \label{fig:policy_price_caps}
\end{figure}


In addition to impacts on capacity, the price cap also biases the operating
behavior of storage under uncertainty. The operating policy is shifted away
from charge and towards discharge as the cap is lowered.
Figure~\ref{fig:policy_price_caps} illustrates this effect by comparing the
optimal dynamic storage policy under VoLL pricing and under a binding price
cap. As the cap is lowered, the state space in which charging is
optimal contracts, while the region in which discharging is optimal  expands.

This distortion arises because the price cap reduces the continuation value of
stored energy and weakens the incentive to charge in anticipation of future
scarcity events. The distortion is similar to the problem identified in
\cite{suskiMissingMoneyMarketBased2025}, whereby a price cap in combination
with storage's need to break even on round-trip efficiency losses can,
under certain net load sequences, remove the incentive to store energy for
load shed events.
\footnote{For example, assume there is a high-price hour immediately followed
by a scarcity hour where the price cap is binding. If the first hour has
a price equal to \$500/MWh, and the storage has $\eta=25\%$, a
cap below \$2,000/MWh removes the incentive to charge in the first hour and
creates more load shedding.}
While \cite{suskiMissingMoneyMarketBased2025} focuses on this
phenomenon in a perfect-foresight setting, we identify here an additional
mechanism
which is not dependent on the round-trip efficiency and is unique
to the dynamic operation. As Table~\ref{tab:price_cap_reliability}
in the Appendix illustrates,
it is possible for the price cap to have no impact on unserved energy under
perfect-foresight operation while having a significant impact in the dynamic
operation by shifting the optimal policy towards earlier discharge in
moderately high-demand states (Fig. \ref{fig:policy_price_caps}).

Capacity payments are a common mechanism used to address this problem by
providing additional payments to resources based on an \textit{ex ante}
calculation of their firm capacity contribution. However, the fact that
storage's preferred operating states are affected by the missing money problem
implies that state-independent capacity payments alone are insufficient to
fully remedy underinvestment in storage capacity. While a fixed capacity
payment will incentivize additional storage investment by offsetting fixed
costs, it cannot address the reduction of the incentives for storage to
hold energy in reserve. Thus, without a state-contingent contract
structure that compensates reserving energy for possible scarcity
events, the distortion to storage operation will create lower revenues
and lead to less investment than would be socially optimal.

\section{Conclusions} \label{sec:conclusion}

This paper examined how demand uncertainty alters the operation and investment
incentives of energy storage in electricity systems. By modeling storage
operation as an infinite-horizon stochastic control problem and embedding the
resulting stationary policies into a capacity expansion framework, we showed
that uncertainty induces a precautionary operating policy where storage
reserves energy to hedge against potential future scarcity events.
This behavior reshapes post-storage net demand
distributions and affects the equilibrium capacity mix
relative to a perfect-foresight benchmark, with storage becoming significantly
less attractive under uncertainty.

A key limitation of our work is that the capacity outcomes are obtained using
a fixed-point algorithm that cycles between dynamic storage operation and
investment, but the algorithm does not guarantee convergence to a globally
optimal solution in the presence of nonconvexities. However, the qualitative
mechanisms identified in this work arise from the structure of the underlying
stochastic control problem and as such are robust to any numerical
approximations or optimality gaps. Developing algorithmic approaches with
stronger optimality guarantees remains an important direction for future
research.

Finally, the reliability externality inherent to electricity markets interacts
with demand uncertainty to uniquely distort storage outcomes when compared to
conventional generators. Price caps and the ``missing money'' problem affect
not only the revenues during scarcity events, but also the structure of the optimal
operating policy by reducing the continuation value of holding energy for
uncertain future scarcity states. These operational distortions impact long-run
planning outcomes, leading to underinvestment in storage capacity that cannot
be corrected by uniform, state-independent capacity payments alone. Our work
highlights the need for capacity contracting mechanisms for resource adequacy
that account for the state-contingent value of stored energy in systems with
high uncertainty and a growing reliance on storage.

\section{AI Usage Disclosure}

ChatGPT was used to assist with grammar, phrasing, and formatting in all
sections of this paper. Code was written with
GitHub Copilot and Claude Code. All AI-generated content was reviewed and
edited for accuracy and clarity.

\appendix
\raggedbottom
\subsection{System Parameters}
\label{appendix:system-values}

{\renewcommand{\arraystretch}{1.1}
\begin{table}[H]
\centering
\begin{threeparttable}
\caption{Stochastic Process Parameters and Demand Discretization}
\label{tab:demand_parameters}
\begin{tabular*}{\columnwidth}{@{\extracolsep{\fill}}lcc@{}}
\toprule
\textbf{Category} & \textbf{Parameter} & \textbf{Value} \\
\midrule
\multirow{3}{*}{OU Process\tnote{*}}
& Long-run mean $\mu$ (GW) & $6.5$ \\
& Volatility $\sigma$ (GW) & $2.29$ \\
& Drift speed $\theta$ ($\text{h}^{-1}$) & $0.0637$ \\
\addlinespace[0.5em]
\multirow{3}{*}{Demand Discretization}
& Minimum load (GW) & $-10$ \\
& Maximum load (GW) & $26$ \\
& Discrete demand states & $71$ \\
\bottomrule
\end{tabular*}
\begin{tablenotes}[flushleft]
\footnotesize
\item[*] The volatility and drift match the autocorrelation and standard
         deviation of the 2024 CAISO net load data; the resulting OU process
         has a standard deviation of 6.4~GW. Independently of the CAISO data,
         the process mean $\mu$ is set to 6.5~GW\@. This places
         a zero-value of net demand approximately one standard deviation below the mean
         such that the process includes a significant amount
         of negative net-demand states.
\end{tablenotes}
\end{threeparttable}
\end{table}
}

{\renewcommand{\arraystretch}{1.1}
\begin{table}[H]
\centering
\begin{threeparttable}
\caption{Generation and Storage Cost Assumptions}
\label{tab:tech_costs}
\begin{tabular*}{\columnwidth}{@{\extracolsep{\fill}}lccc@{}}
\toprule
\textbf{Technology}
& \textbf{Fixed Cost\tnote{*}}
& \textbf{Variable Cost}
& \textbf{Efficiency} \\
& (\$/MW-yr)
& (\$/MWh)
& \\
\midrule
Baseload
& 400{,}000
& 12
& -- \\
Peaker
& 110{,}000
& 180\tnote{\textdagger}
& -- \\
\midrule
Battery Storage (4h)
& 65{,}000
& --
& 0.85 \\
\bottomrule
\end{tabular*}
\begin{tablenotes}[flushleft]
\footnotesize
\item[*] Fixed cost values are loosely based on the NREL Annual Technology
      Baseline for 2050.
\item[\textdagger] The variable cost of the peaker is deliberately set at a
      relatively high level to induce storage investment in our examples.
\end{tablenotes}
\end{threeparttable}
\end{table}
}

\subsection{Reliability Statistics Comparison}\label{appendix:reliability}

{\renewcommand{\arraystretch}{1.1}
\begin{table}[H]
\centering
\begin{threeparttable}
\caption{Reliability Indices Under Varying Price Caps\tnote{*}}
\label{tab:price_cap_reliability}
\begin{tabular*}{\columnwidth}{@{\extracolsep{\fill}}lcccc@{}}
\toprule
\textbf{Reliability Index \tnote{\textdagger}} & \textbf{VoLL (\$/MWh)} &
\multicolumn{3}{c}{\textbf{Price Cap (\$/MWh)}} \\
\cmidrule(lr){2-2}\cmidrule(lr){3-5}
& 10{,}000 & 5{,}000 & 2{,}500 & 1{,}000 \\
\midrule
\multicolumn{5}{@{}l}{\textit{Dynamic Operation}} \\
LOLE \tnote{\textdaggerdbl} & 33.3 & 36.0 & 34.3 & 33.8 \\
EUE  & 39.4 & 43.9 & 49.2 & 58.7 \\
\addlinespace[0.35em]
\multicolumn{5}{@{}l}{\textit{Perfect-Foresight Operation}} \\
LOLE & 11 & 11 & 11 & 11 \\
EUE & 28.5 & 28.5 & 28.5 & 28.5 \\
\bottomrule
\end{tabular*}

\begin{tablenotes}
\footnotesize
\item[*] All operations are determined with the \textit{Dynamic Operation} + VoLL capacities in Table~\ref{tab:price_cap_capacity}.
      Reliability indices are evaluated over a 10-year
      sample from the OU process at hourly resolution.
\item[\textdagger] Loss of load expectation (LOLE) is in hours / year;
    expected unserved energy (EUE) is in GWh / year.
\item[\textdaggerdbl] LOLE for the \textit{Dynamic Operation} case does not
    change monotonically; we hypothesize that this is due to degeneracy
    in how load can be shed, \textit{e.g.} shedding 1 MW of load for
    1 hour has an equivalent system cost as shedding 0.5 MW of load for 2 hours.
\end{tablenotes}
\end{threeparttable}
\end{table}
}

\vspace{10em}
\bibliographystyle{ieeetr}
\bibliography{references}

\endgroup
\end{document}